%% file: 0_main_hal.tex
\documentclass[a4paper,UKenglish, autoref, thm-restate]{article}

\usepackage[utf8]{inputenc}
\usepackage{xcolor}
\usepackage{graphicx}
\usepackage{xspace}
\usepackage{url}
\usepackage{amsmath,amssymb,bm,amsfonts,amsthm}
\usepackage{hyperref}
\usepackage{authblk}
\usepackage{cleveref}



\usepackage[vlined,english,boxed,linesnumbered,ruled]{algorithm2e} 
      \IncMargin{1em}  
     \SetKwInput{Input}{Input}
     \SetKwInput{Output}{Output}
     \SetKw{Return}{Return}
     \SetKwIF{If}{ElseIf}{Else}{If}{then}{else if}{else}{endif}
     \SetKwFor{For}{For}{do}{done}
     \SetKwFor{for}{For}{}{}
     \SetKwFor{While}{While}{do}{done}
     \SetKwFor{Repeat}{Repeat}{}{}
     \SetKwRepeat{Do}{Do}{while}
      \SetKwIF{Do}{notused}{WhileOfDo}{Do}{}{notused}{while}{done}
     \SetKwFor{Procedure}{Procedure}{}{}
     \SetKwFor{Function}{Function}{}{}
     \SetKw{Return}{Return}%
     \SetKw{return}{return}%
     \SetKwComment{Comment}{/* }{ */}
     \SetSideCommentRight
     \DontPrintSemicolon

\SetCommentSty{xCommentSty}

\makeatletter
\newcommand*{\algotitle}[2]{%
  \stepcounter{algocf}%
  \hypertarget{algocf.title.\theHalgocf}{}%
  \NR@gettitle{#1}%
  \label{#2}%
  \addtocounter{algocf}{-1}%
}
\makeatother

\newcommand{\set}[1]{\{#1\}}
\newcommand{\vcdim}{\mathrm{VCdim}}
\newcommand{\floor}[1]{\lfloor #1\rfloor}
\newcommand{\card}[1]{\left\| #1\right\|}
\newcommand{\cH}{\mathcal{H}}
\newcommand{\R}{\mathcal{R}}
\renewcommand{\P}{\mathcal{P}}
\def\bigO{\ensuremath{\mathcal{O}}\xspace}

\newcommand{\graph}[1]{{#1}}

\newtheorem{theorem}{Theorem}

\newtheorem{lemma}{Lemma}
\newtheorem{corollary}{Corollary}

\newtheorem*{claim*}{Claim}

\title{Practical Computation of Graph VC-Dimension\footnote{This work was supported by a grant of the Romanian Ministry of Research, Innovation and Digitalization, CCCDI - UEFISCDI, project number PN-III-P2-2.1-PED-2021-2142, within PNCDI III, and the French National Research Agency (ANR) through project Tempogral with reference number ANR-22-CE48-0001.}} 

\author[1]{David Coudert}
\affil[1]{Université Côte d'Azur, Inria, I3S, CNRS, France}
\author[2]{Mónika Csikós}
\affil[2]{IRIF, CNRS and Université Paris Cité, Paris, France}
\author[3]{Guillaume Ducoffe}
\affil[3]{University of Bucharest, Romania \& National Institute for Research and Development in Informatics, Romania}
\author[4]{Laurent Viennot}
 \affil[4]{Inria, DI ENS, Paris, France}

\begin{document}

\maketitle
\begin{abstract}
For any set system $\cH=(V,\R), \ \R \subseteq 2^V$, a subset $S \subseteq V$ is called \emph{shattered} if every $S' \subseteq S$ results from the intersection of $S$ with some set in $\R$. The \emph{VC-dimension} of $\cH$ is the size of a largest shattered set in $V$. In this paper, we focus on the problem of computing the VC-dimension of graphs. In particular, given a graph $G=(V,E)$, the VC-dimension of $G$ is defined as the VC-dimension of $(V, \mathcal N)$, where $\mathcal N$ contains each subset of $V$ that can be obtained as the closed neighborhood of some vertex $v \in V$ in $G$. 
Our main contribution is an algorithm for computing the VC-dimension of any graph, whose effectiveness is shown through experiments on various types of practical graphs, including graphs with millions of vertices. A key aspect of its efficiency resides in the fact that practical graphs have small VC-dimension, up to 8 in our experiments. As a side-product, we present several new bounds relating the graph VC-dimension to other classical graph theoretical notions. We also establish the $W[1]$-hardness of the graph VC-dimension problem by extending a previous result for arbitrary set systems.
\end{abstract}

\textbf{Keywords:} {VC-dimension, graph, algorithm}

\clearpage
\input{1_introduction.tex}

\input{2_definitions.tex}

\input{3_hardness.tex}

\input{4_simple_bounds.tex}
\input{5_algorithm.tex}

\input{6_experiments.tex}

\input{8_conclusion.tex}

\bibliography{biblio}

\appendix
\input{9_appendix.tex}

\end{document}

%% file: 1_introduction.tex
\section{Introduction}
Since the seminal work of Vapnik and Chervonenkis~\cite{doi:10.1137/1116025}, \emph{VC-dimension} is one of the basic quantities describing the complexity of a set system. As such, VC-dimension is the foundation of many results in mathematics and theoretical computer science: it plays a central role in uniform sampling guarantees and is often required as an input parameter of algorithms.
For instance, one of the fundamental  results in computational learning theory states that a set system is PAC-learnable if and only if it has bounded VC-dimension~\cite{blumer1989learnability}.
The Sample Compression Conjecture, which has been called in~\cite{chalopin2022unlabeled} one of the oldest open problems in theoretical machine learning, is also related to VC-dimension~\cite{floyd1995sample}.
VC-dimension has become a key concept in other fields as well. In computational geometry, bounds on the VC-dimension of various geometric set systems ({\it e.g.}, ones induced by half-spaces, balls or axis-parallel boxes) are essential parameters in methods developed for approximating and processing geometric data~\cite{despres2017vapnikchervonenkis,wenocur1981some}.
Indeed, set systems with bounded VC-dimension admit structures such as small-size $\varepsilon$-nets~\cite{haussler1986epsilon} and $\varepsilon$-approximations~\cite{doi:10.1137/1116025, LiLS-sample-complexity-learning-S01, csikos2022optimal}, matchings and spanning paths of low crossing numbers~\cite{welzl1988partition,chazelle1989quasi}, colorings with low discrepancy~\cite{Mat99} to name a few.
We refer to the survey~\cite{matouvsek1998geometric} for more details.

Practical applications of this parameter (e.g. in the design of PAC-learning algorithms), require bounds on the VC-dimension of the considered set systems.
However, it was observed that the general bounds found in the literature are of limited use in practice~\cite{holden1995practical}.
Thus, the problem of computing the VC-dimension of set systems has attracted some attention.
It is proved to be $\log$NP-hard~\cite{papadimitriou1996limited}, and $W[1]$-hard for the natural parameterization by the VC-dimension~\cite{downey1993parameterized}.
Furthermore, under plausible complexity hypotheses, the VC-dimension is hard to be approximated within a sub-logarithmic factor in polynomial time~\cite{manurangsi2017inapproximability}.

\smallskip
There is a fast growing body
of literature demonstrating the strong potential of using VC-dimension as a graph  parameter. In recent studies, researchers have made
significant progress in improving results in extremal and algorithmic graph theory by limiting the problems to
objects with bounded VC-dimension \cite{LT10,FPS21,ducoffe2021computing,ducoffe2022diameterb,bousquet2015identifying,bonamy2021eptas}, including an EPTAS for the {\sc Maximum Clique} problem~\cite{bonamy2021eptas} and subquadratic-time algorithms for diameter computation~\cite{ducoffe2021computing,ducoffe2022diameterb}. The VC-dimension of a graph was also linked to the complexity of approximating a minimum-cardinality identifying code on hereditary graph classes~\cite{bousquet2015identifying}.

The most commonly used notion of VC-dimension for graphs is defined as the VC-dimension of its neighbourhood set system (see  Section~\ref{sec:definition} for a formal definition). 
More specifically, in this paper we consider the \emph{closed} neighbourhoods of vertices (\textit{i.e.}, each vertex is included in the set of its adjacent vertices).
However, we could instead consider their \emph{open} neighbourhoods.
Both notions result in different but comparable values for the VC-dimension.
The algorithmic applications listed in~\cite{bousquet2015identifying,ducoffe2021computing,ducoffe2022diameterb} are proved using the set systems of closed neighbourhoods in a graph, while those listed in~\cite{bonamy2021eptas,FPS21} are proved using the set systems of open neighbourhoods.
Note that the VC-dimension of other graph-related set systems has been considered in~\cite{chepoi2007covering,chepoi2020density,ducoffe2022diameterb,kranakis1997vc}, with different combinatorial and algorithmic implications.
These alternative VC-dimension parameters are {\em not} considered in our paper. However, it is noteworthy that several of these parameters can be lower bounded by the graph VC-dimension.
More generally, it was observed in~\cite{chalopin2023sample} that every set system $\mathcal{H}$ can be represented as a split graph $G_{\mathcal{H}}$, in such a way that the VC-dimension of $\mathcal{H}$ is equal to the VC-dimension of the closed neighbourhoods of vertices in the stable set of $G_{\mathcal{H}}$. Therefore, graph VC-dimension is as general as the VC-dimension of arbitrary set systems, if we allow ourselves to only consider the closed neighbourhoods of a restricted subset of vertices.


Just like for general set systems, the problem of computing the VC-dimension of a graph is known to be  $\log$NP-hard~\cite{kranakis1997vc}.
However, we are not aware of any previous study on the parameterized complexity of the problem. Similarly, very few is known about the VC-dimension of complex networks. The closest such related work would be~\cite{demaine2019structural}, where the stronger property of bounded expansion is considered.
The VC-dimension of random graphs has been studied in~\cite{anthony1995vapnik}, where for any fixed value $d$, a density threshold for the property of having VC-dimension at most $d$ is derived.

\paragraph*{Our Contributions.} 

While developing improved methods for graphs with bounded VC-dimension is a fruitful direction, it is just as important to provide efficient algorithms and conditions to help computing or approximating the VC-dimension of the input graph.
We address this problem  both from a theoretical and practical point of view.

The \emph{main contribution} of this paper is a practical algorithm for computing the VC-dimension of any graph (Algorithm~\ref{alg:exact}).
Note that a naive algorithm for this problem would consider all vertex subsets of size at most the VC-dimension of the input. 
By contrast, our algorithm repeatedly updates a lower bound on the VC-dimension, so that most unexplored vertex subsets below this bound can be discarded.
Furthermore, while exploring for larger shattered subsets, we use our degree-based upper bounds on the VC-dimension in order to discard at once all vertices of too small degree (at most exponential in the current lower bound).
Similarly, we show that while growing a shattered subset by iteratively adding new vertices, some branches can be ignored using a simple, but surprisingly powerful, upper bound on the size of a largest shattered superset (Lemma~\ref{lem:base}).   
By doing so, we considerably reduce the search space, as evidenced by our experiments on some real-life networks.
We implemented a few more tricks, based on a combination of bit masks and partition refinement techniques, in order to speed up some important routine tasks in the algorithm, such as: the test of whether a given subset is shattered, that of whether a search branch can be pruned, and reduction schemes for the graph to be considered.
Overall, we were able to compute the VC-dimension of graphs with millions of nodes in less than 40 minutes, providing the first practical algorithm for computing graph VC-dimension.
We demonstrate the efficiency of our algorithm on various practical graphs.
To the best of our knowledge, this is the first analysis of the VC-dimension in real networks.
Interestingly, we observe that for all graphs considered in our experiments, the VC-dimension ranges between $3$ and $8$.

\smallskip
Our next contribution is proving that computing the VC-dimension of a graph is a $W[1]$-hard problem for the natural parameterization by the VC-dimension (Theorem~\ref{thm:w1-hard}).
For that, we revisit a previous $W[1]$-hardness proof for arbitrary set systems~\cite{downey1993parameterized}, which we combine with some insights on shattered subsets in graphs from~\cite{ducoffe2022diameter}.

Finally, we note that we obtain a series of new bounds on the VC-dimension with respect to classical graph parameters, such as maximum degree, degeneracy and matching number.
Some of these parameters are included in the setup of ~\cite{bousquet2015identifying} who show that the VC-dimension of a graph can be functionally upper bounded by any hereditary graph parameter that stays unbounded on the following graph classes: split graphs, bipartite graphs and co-bipartite graphs.
However, the bounds that can be derived from~\cite{bousquet2015identifying} are rather rough, due to the use of Ramsey's theory.
By contrast, we give linear and sharp bounds for all the considered parameters.

\smallskip
\textbf{Organization. } After defining the main notions and notation of this paper in Section~\ref{sec:definition}, we prove that computing the VC-dimension of graphs is $W[1]$-hard in Section~\ref{sec:hardness}. Then in Section~\ref{sec:bounds}, we summarise our bounds on the graph VC-dimension.
In Section~\ref{sec:algorithm}, we give a new exact algorithm for computing the VC-dimension of graphs and discuss several possible optimizations. 
In Section~\ref{sec:experiments}, we report our experimental results 
and discuss the advantages of the different optimizations methods. 

%% file: 2_definitions.tex
\section{Definitions and notation}
\label{sec:definition}

Throughout this note, we use lowercase letters $u,v,x,y,\ldots$ for vertices, uppercase letters $X,Y,Z,\ldots$ for sets of vertices,  calligraphic letters as $\R$ for collections of sets, and we let $\log$ denote the base 2 logarithm.
Given an undirected graph $G=(V,E)$ with $|V|=n$ vertices and $|E|=m$ edges,
let $N_G[v]$ denote the closed neighborhood of $v$ defined as $N_G[v]=\set{u\in V \mid uv\in E \mbox{ or } u= v}$.
We define the degree $\deg(v)=|N_G[v]|$ of a vertex as its closed neighborhood cardinality. We use this unusual convention of counting a vertex in its own degree for the sake of simplicity when considering closed neighborhood sizes. We also define the ball $B_G[v,r]$ centered at a vertex $v$ and with radius $r$ as the set of nodes at distance at most $r$ from $v$. In particular, we have $B_G[v,1]=N_G[v]$.
We omit the ${}_G$ subscript when $G$ is clear from the context.

A set system $\cH=(V,\R)$ (or hypergraph) is defined by a ground set $V$ and a collection $\R$ of subsets of $V$ called ranges.
Recall that a set $X\subseteq V$ is said to be \emph{shattered} by $\R$ (or simply shattered if $\R$ is clear from the context) if for every $Y \subseteq X$ there exists a range $R\in \R$ such that $Y=R\cap X$. For any $R\in \R$, the intersection $Y=R\cap X$ is called the \emph{trace} of $R$ on $X$. The VC-dimension of a graph $G$, denoted by $\vcdim(G)$, is defined as the VC-dimension of its closed neighborhood set system $(V,\set{N_G[v] \mid v\in V})$. A subset $X\subseteq V$ is thus shattered if for every $Y \subseteq X$ there exists a vertex $v_Y$ such that its trace $N_G[v_Y] \cap X$ on $X$ equals $Y$. {When considering such a set system, we say that a vertex $v$ has trace $Y$ on a set $X$ when its closed neighborhood has trace $Y=N_G[v]\cap X$.}

%% file: 3_hardness.tex
\section{W[1]-hardness}
\label{sec:hardness}
As we have mentioned in the introduction, computing the VC-dimension of graphs is known to be $\log$NP-hard \cite{kranakis1997vc}. We show that it is also W[1]-hard for the natural parameterization by the VC-dimension by showing the following statement in Appendix~\ref{sec:proof-w1-hard}. 

\begin{theorem}\label{thm:w1-hard}
For any graph $G$ and parameter $ k \leq |V(G)|$, there exists a graph $H_G$ such that $G$ contains a $k$-clique if and only if the VC-dimension of $H_G$ is at least $k$.
Furthermore, we can construct $H_G$ from $G$ in $\bigO(k2^kn^2)$ time.
\end{theorem}

%% file: 4_simple_bounds.tex
\section{Simple bounds}
\label{sec:bounds}

The following are upper bounds on the size of shattered subsets in graphs, with respect to various graph parameters.
We emphasize on Lemma~\ref{lem:base}, which is a cornerstone of our practical algorithm for computing the VC-dimension of graphs (presented in the next Section~\ref{sec:algorithm}).

\begin{lemma}\label{lem:base}
  Consider a shattered set $X$ and $Y\subseteq X$.
  Let $Y'$ be the set of vertices with trace $Y$ on $X$.
  Any shattered set $Z$ containing $X$ satisfies $2^{|Z|-|X|}\le |Y'|$.
\end{lemma}

\begin{proof}
  For any subset $X'\subseteq Z\setminus X$, there must exist a vertex $v_{X'}$ with trace $N[v_{X'}]\cap Z=Y\cup X'$.  As $v_{X'}$ has trace $Y$ on $X$, it is included in $Y'$. All vertices $v_{X'}$ are pairwise distinct since they have pairwise distinct traces on $Z$. Hence, $Y'$ has size at least $2^{|Z\setminus X|}$.
\end{proof}

\noindent
Setting $X=Y=\set{x}$ for any vertex $x$, we have $Y'=N[x]$ and obtain the following bound.

\begin{corollary}\label{cor:degbound}
Any shattered set $Z$ containing a vertex $x$ has size at most $\floor{\log \deg(x)} + 1$. 
\end{corollary}

In Appendix~\ref{sec:more-simple-bounds}, we additionally relate the VC-dimension of a graph with its degeneracy $k$ and its matching number $\nu$ (recall that a graph $G$ is called $k$-degenerate if every subgraph of $G$ contains a vertex with at most $k$ neighbours, and that $\nu$ is the  size of a maximum matching in $G$).
We summarize the obtained upper-bounds in the next lemma.
\begin{lemma}\label{lem:upper-bounds}
    Let $G$ be a non-empty graph on $n$ vertices with maximum degree $\Delta$, matching number $\nu$, and degeneracy $k$, then 
    $$
        \vcdim(G) 
        \leq 
        \min \left\{ 
        \lfloor \log n \rfloor, \lfloor \log \Delta \rfloor + 1, k+1, \nu +1
        \right\}.
    $$
\end{lemma}

The following lemma allows to restrict the search of a shattered set containing a given node $x$ to its ball $B[x,2]$ of radius 2.
 \begin{lemma}\label{lem:ball-restriction}
For any shattered set $X$ and $x \in X$, we have $X \subseteq B[x,2]$.
\end{lemma}
\begin{proof}
    Since $X$ is shattered, for any vertex $y \in X \setminus \{x\}$, there exists a vertex $v \in V$ such that $N[v] \cap X = \{x,y\}$. That is, $v$ is a common neighbor of $x$ and $y$ and so $\mathrm{dist}_G(x,y) \leq 2$.
\end{proof}

%% file: 5_algorithm.tex
\section{Algorithm}
\label{sec:algorithm}

The exact algorithm is given in Algorithm~\ref{alg:exact}. Apart from the graph, it receives a lower bound $lb$ on its VC-dimension. Alternatively, one can start the algorithm with input value $lb = 0$. In Section~\ref{sec:LB}, we describe a way to obtain a better starting value for $lb$.


\subsection{Outline of the method}
Given a graph $G$ and a lower bound $lb$ of its VC-dimension, our algorithm consists in exploring all possible shattered sets of size at least $lb+1$ according to Algorithm~\ref{alg:exact}. If one is found, $lb$ is updated and the search is continued on larger sets. When no shattered set of size $lb+1$ is found, we conclude that $lb$ is equal to the VC-dimension. The most technical part---checking for shattered supersets---is contained in the function \nameref{ShatAlgo} (see Algorithm~\ref{subalgo:explore}). We encode a subset $Y\subset X=\set{x_1,\ldots,x_k}$ by the integer with binary representation $y = y_{k}\cdots y_1$ where $y_i=1$ if $x_i\in Y$ and $y_i=0$ otherwise. The trace $N[v]\cap X$ of each vertex $v$ is therefore stored in a mask $M[v]$ which is updated as we visit subsets $X$ of $H$: the $i$th bit of $M[v]$ indicates whether the $i$th vertex of $X$ is in $N[v]$.

\begin{algorithm}[ht]
\algotitle{\textsc{VCdimComputation}}{VCAlgo}
\caption{  \textsc{VCdimComputation}$\big( G, n, lb\big)$}  \label{alg:exact} \small
  \Input{A graph $G=(V,E)$ with $n=|V|$ vertices, lower bound $lb$ on $\vcdim(G)$}
  \Output{The VC-dimension of $G$}
  \smallskip
  Let $H$ be an array containing all vertices of degree at least $2^{lb}$\\ 
  Sort $H$ (optional). \tcp*{We consider 3 different ways of sorting, see Section~\ref{sec:optimization}} 
   Initialize a mask $M[v]\leftarrow 0$ for all $v\in H$. \tcp*{Trace of $N[v]$ on $X=\emptyset$}
    Initialize $T = [T[0]]\leftarrow [n]$ \\
     \tcp{$T[y]$ counts the number of vertices $v\in V$ with trace $M[v]=y$ on $X =\emptyset$}
    \smallskip
  \For{$i = 1$ to $|H|$ }{
           
    $lb \leftarrow $ \nameref{ShatAlgo}$(H,i,\emptyset,T,lb)$\label{lin:forH}
  }
  \smallskip
  \Return $lb$ 
  
\end{algorithm}

\begin{algorithm}[ht]
\algotitle{\textsc{ExploreShattered}}{ShatAlgo}
\caption{  \textsc{ExploreShattered}$\big( H,i,X,T,lb\big)$} \small
    \label{subalgo:explore}

    
      Set $x\leftarrow H[i]$, $s\leftarrow  |X \cup \{x\}|$, and $m\leftarrow 2^{s-1}$. \tcp*{$m$ is the bit mask for $x$}

        \smallskip 
        
        $T' \leftarrow \nameref{traceAlgo} \big(T,x,m \big)$ 
        \tcp*{$T'[y]  =$ \# vertices with trace $y$ on $X \cup \{x\}$ }



        \smallskip
      
      $prune\leftarrow \texttt{False}$\\
      \For{$y=0$ to $2^{s}-1$}{
        \lIf{$T'[y] < 2^{lb+1-s}$}{$prune\leftarrow \texttt{True}$. \tcp*[f]{by Lemma~\ref{lem:base}}}
      }

    \smallskip
      
      \If{\emph{\textbf{not}} $prune$}{
      \For{$v\in N_G[x]$}{ 
        $M[v] \leftarrow M[v] + m$\tcp*{Update $M(v)$ to be the trace of $N[v]$ on $X\cup\set{x}$}}
        \lIf{$s>lb$}{$lb\leftarrow s$ \tcp*[f]{$X\cup\set{x}$ is shattered}}
        \For{$j=i+1$ to $|H|$}{\label{lin:forHexplore}
        $lb \leftarrow \nameref{ShatAlgo}(H,j,X\cup\set{x},T',lb)$\\  
        }
         \lFor{$v\in N_G[x]$}{$M[v] \leftarrow M[v] - m$ \tcp*[f]{Restore the trace of $N[v]$ on $X$}}
      }
      \smallskip
    \Return $lb$

    \bigskip

    \Function{\textsc{TraceCountAdd}$\big(T,x,m \big)$ \algotitle{\textsc{TraceCountAdd}}{traceAlgo}}{

      Let $T'$ be a copy of $T$ resized to $2^{s}$ and padded with zeros. 

      \For{$v\in N_G[x]$}{ 
        Consider  $y\leftarrow M[v]$ \tcp*{The trace of $N[v]$ on $X$}
        $T'[y]\leftarrow T'[y]-1$\\
        
        $T'[y+m]\leftarrow T'[y] + 1$\\
      }
    
    \Return $T'$
  }
  
\end{algorithm}




        
    

\noindent
To make the  algorithm  more efficient, we incorporated the following key ideas:
\begin{itemize}
    \item According to Corollary~\ref{cor:degbound}, we can restrict the search to the set $H$ of vertices with degree $2^{lb}$ at least. 
    \item We fix an ordering $\prec$ of $H$ and scan subsets of $H$ in a depth first search manner.\\
    More precisely, \nameref{ShatAlgo} performs a DFS of the inclusion graph of subsets of $H$ by following arcs $X\rightarrow Z$ for $X,Z\subseteq H$ such that $Z=X\cup\set{z}$ and $x\prec z$ for all $x\in X$. Starting from the empty set, any set $X=\set{x_1,\ldots,x_k}$ with $x_1\prec\cdots\prec x_k$ is thus reachable through $\emptyset\rightarrow \set{x_1}\rightarrow \set{x_1,x_2}\rightarrow\cdots\rightarrow X$.
    \item For each visited set $X$, we compute a table $T$ counting for each $Y\subseteq X$ the number of vertices $v$ with trace $N[v]\cap X=Y$. If $T[Y]<2^{lb+1-|X|}$ for some $Y$, then Lemma~\ref{lem:base} implies that there exists no shattered set $Z\supseteq X$ of size $lb+1$ or more, so we do not explore the supersets of $X$. Note that this test is  not satisfied when $X$ is not shattered as we then have $T[Y]=0$ for some $Y\subseteq X$. The argument $X$ of \nameref{ShatAlgo} is thus always a shattered set.
    \item 
    When considering $Z=X\cup\set{x}$, the table $T'$ for $Z$ can be obtained from $T$ in time $\bigO(|T|+\Delta)=\bigO(2^d+\Delta)$ where $\Delta$ is the maximum degree in $G$ and $d$ is its VC-dimension. 
\end{itemize}

 The overall worst case complexity of the algorithm is thus $\bigO(n^d(2^d+\Delta))$ as we visit only shattered sets. Moreover, Corollary~\ref{cor:degbound} implies $2^d=\bigO(\Delta)$ and the complexity is thus $\bigO(n^d\Delta)$. This pessimistic bound assumes that a constant fraction of sets of size $d$ are shattered. Empirically we observed much better running times, as reported in Section~\ref{sec:experiments} with the evaluation on practical graphs.
Section~\ref{sec:optimization} describes several optimizations, including how to sort $H$.

\subsection{Lower-bound computation}\label{sec:LB}

We compute a lower bound $lb$ by a method similar to \nameref{VCAlgo} with $lb = 0$ (thus $H=V$). We make the search faster by only performing a partial scan of subsets of $H$. In order to find large shattered sets, we sort $H$ by non-increasing degrees. We restrict the search in two ways: we fix a maximum number $maxvisits=64$ of times a vertex $x$ can be added to the current shattered set; we also restrict the for loop of Line~\ref{lin:forHexplore} of \nameref{ShatAlgo} to the first $maxvisits/2$ elements. 
As each vertex is visited a constant number of times (at most $maxvisits$), this modified version takes linear time.

\subsection{Optimizations}
\label{sec:optimization}


\paragraph*{Vertex ordering.} 
We have considered the following options for sorting the set $H$ of high degree vertices:
\begin{itemize}
  \item non-increasing degrees ($D^-$),
  \item non-decreasing degrees ($D^+$),
  \item $k$-core ordering  ($K$): an ordering of $G_{|H}$, the subgraph restricted to $H$, obtained by repeatedly removing a vertex with lowest degree, vertices removed first are ordered first,
  \item random ordering ($R$).
\end{itemize}
The intuition for choosing non-increasing degrees follows that of the lower bound computation: higher degree vertices tend to participate to larger shattered set and exploring them first can improve the lower bound earlier, allowing to restrict the rest of the search more severely thanks to Lemma~\ref{lem:base}. Conversely, if we start with a good enough lower bound, lower degree vertices tend to participate to smaller shattered sets and the exploration from these vertices tends to be faster. Exploring them first then speeds up the exploration from high degree vertices as we do not need to consider adding already scanned vertices anymore. Using a $k$-core ordering follows a similar intuition with the refinement of taking into account the degree after removing previous vertices rather than the degree in the full graph.
Using a random ordering seems a basic choice for comparison.

\paragraph*{Ball restriction.} 
Lemma~\ref{lem:ball-restriction} implies that when starting an exploration from $x = H[i]$ (Line~\ref{lin:forH} of \nameref{VCAlgo}) we can restrict the search to consider only vertices in $ B[x,2]\cap\set{H[i],H[i+1],\ldots}$.

\paragraph*{Graph reduction.}
As we focus on shattered sets included in $H$, we can restrict the graph while preserving all possible traces on $H$. For that purpose, for each possible trace $Y\subseteq H$ which can be obtained as $H\cap N[v]$ for $v\in V$, we keep at most one vertex $v\in V\setminus H$ with trace $N[v]\cap H=Y$. Such a selection can be obtained in linear time using partition refinement~\cite{PT87,HMPV00} as follows. Starting from the partition $\P=\set{V}$ we iteratively refine it by sets $N[x]$ for $x\in H$: each refinement step consists in splitting each part $P\in \P$ into $P\cap N[X]$ and $P\setminus N[X]$ (if one of the two sets is empty, $P$ remains unchanged). Each refinement step clearly maintains the invariant that all vertices in a part have the same trace on the set of vertices of $H$ processed so far. At the end of the process, all vertices in a part must have the same trace on $H$. We thus keep one vertex per final part not intersecting $H$ and all vertices in $H$ to obtain a set $V'\supseteq H$ of vertices providing the same traces on $H$ as $V$ and proceed on $G_{|V'}$ instead of $G$.

%% file: 6_experiments.tex
\begin{table}[t]
  \caption{The graphs we use with their main parameters and the time (in seconds) required by our reference implementation $KBG$ to compute their VC-dimension.}
  \label{tab:ref}
\centering
\setlength{\tabcolsep}{4pt}
\begin{tabular}{|l|rrrrr|}
\hline
\textbf{Graph} 
& \textbf{\#nodes} & \textbf{\#edges} & \textbf{max.deg.}
& \textbf{VC-dim $d$} & \textbf{time (s)}  \\
\hline
\hline
\input{tbls/table1.tex}
\hline
\end{tabular} 
\end{table}

\section{Experiments}
\label{sec:experiments}

\subsection{Dataset}

We evaluate the performance of our algorithm on various types of practical graphs. We use
graphs from the BioGRID interaction database (BIO-*)~\cite{BIOGRID};
a protein interactions network (dip20170205)~\cite{DIP}; 
and graphs of the autonomous systems from the Internet (oregon2, CAIDA\_as and DIMES)~\cite{P2P,CAIDA,DIMES}.
We also test 
computer networks (Gnutella, Skitter), 
web graphs (notreDame and BerkStan), 
social networks (Epinions, Facebook, Twitter),
co-author graphs (ca-HepPh, dblp),
road networks (t.CAL, t.FLA), 
a 3D triangular mesh (buddha), 
a graph from a computer game (FrozenSea),
and grid-like graphs from VLSI applications (z-alue7065).
The data is available from \url{snap.stanford.edu}, \url{webgraph.di.unimi.it}, \url{www.dis.uniroma1.it/challenge9}, \url{graphics.stanford.edu}, \url{steinlib.zib.de}, and \url{movingai.com}. 
Furthermore, we use synthetic inputs: grid300-10 and grid500-10 are square grids with respective sizes $301\times 301$, and $501\times 501$ where 10\% of the edges were randomly deleted; and powerlaw2.5 is a random graph generated according to the configuration model with a degree sequence following a power law with exponent 2.5.

All experiments were performed on a cluster of 20 nodes equipped with two Cacade Lake
Intel Xeon 5218 16 cores processors at 2.4GHz and 192GB of memory each. Sixteen processes were run in parallel on four nodes of the cluster. Our computation times are thus pessimistic compared to running each process on a fully reserved node (for example, the computation of the VC-dimension of the twitter graph of our dataset takes 431s on a fully reserved node compared to 456s in our experiment even though we report user times).
The code is available at \href{https://gitlab.inria.fr/viennot/graph-vcdim}{gitlab.inria.fr/viennot/graph-vcdim}.

\subsection{Graphs and reference time}

We first present our dataset graphs in Table~\ref{tab:ref} together with their VC-dimension and time required to compute it with our reference implementation $KBG$: which uses $k$-core ordering ($K$), ball restriction ($B$) and graph reduction ($G$). We have chosen $KBG$ as our reference choice of optimizations as it provides the minimum sum of running times over all graphs of our dataset. We list also their size in terms of number of nodes (not counting isolated nodes) and number of edges (duplicate edges are removed).

We observe that the VC-dimension of all graphs in the dataset is rather small: at most 8, even for graphs with millions of nodes and over 10M edges (M stands for million). Moreover, its computation with our $KBG$ implementation takes at most a few seconds for most of the graphs, and less than a few minutes for all but one: \graph{as-skitter} for which it takes around 35 minutes. Memory usage (not reported here) grows with the input graph size up to roughly 600 megabytes for as-skitter. Not surprisingly, this most difficult graph is both the largest in terms of number of edges and the most complex in terms of VC-dimension.
We analyze the dependency of the computation time with respect to some graph parameters in Section~\ref{sec:analysis}.


\begin{figure}[t]
    \centering
    \includegraphics[width=.95\textwidth]{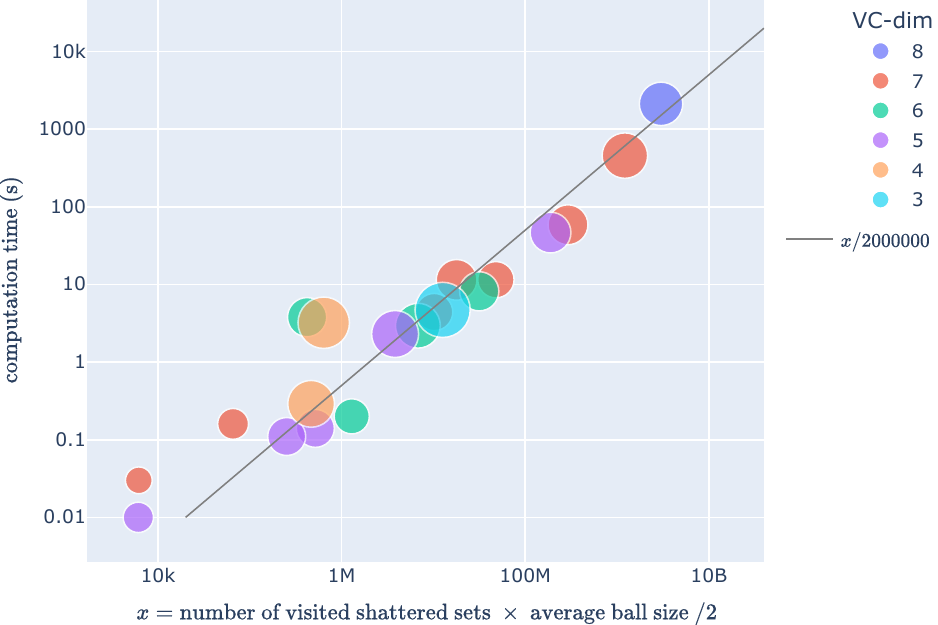}\\
    \caption{Computation time $t$ in seconds versus the estimated number $x$ of tentative shattered sets considered by $KBG$: each network in the dataset is represented by a disk with coordinates $(x,t)$, whose color indicates the VC-dimension $d$ of the network, while the size is proportional to the logarithm of the number of high degree nodes.}
    \label{fig:nshat-compl}
\end{figure}

To appreciate these running times, recall that our algorithm for computing the VC-dimension consists in first computing a lower-bound $lb$, then identifying the set $H$ of \emph{high degree nodes}, that is those with degree at least $2^{lb}$, and then ordering this set for exploring shattered sets included in $H$. 
Figure~\ref{fig:nshat-compl} shows that running times are basically proportional to the number $x$ of shattered sets considered by $KBG$, which is estimated as follows. For each visited shattered set $X=\set{v_1,\ldots,v_k}$ with $v_1\prec\cdots\prec v_k$, our algorithm tries to add high degree nodes of $B[v_1,2]$ coming after $v_k$ in the ordering $\prec$ used for $H$. We can thus estimate $x$ as the product of the number $s$ of visited shattered sets multiplied by half of the average ball size $bsize=\frac{1}{|H|}\sum_{v\in H}\card{B[v,2]\cap H}$. These numbers are reported in Table~\ref{tab:nshat}. We see in Figure~\ref{fig:nshat-compl} that most networks are close to the black line that corresponds to a rate of 2 millions tentative shattered sets per second. For low values of $x$ five networks appear significantly above the line: \graph{p2p-gnutella09}, \graph{BIO-SYS-Aff-Cap-RNA-3.5}, \graph{DIMES-201204}, \graph{powerlaw2.5} and \graph{buddha} (from left to right). This is due to the overhead of reading the input file, computing the lower-bound, the $k$-core ordering, and the graph reduction which appears to be higher than the time for exploring shattered sets in these networks (we get comparable times when running all these phases and stopping  before exploring).

\begin{table}[t]
    \caption{VC-dimension $d$, counts of shattered sets and high degree nodes, average ball size.}
  \label{tab:nshat}
  \small
\centering
\resizebox{\linewidth}{!}{%
\setlength{\tabcolsep}{4pt}
\begin{tabular}{|l|rrrrrrrrr|}
\hline
\textbf{Graph} & $d$
& all & $\deg\ge 2^d$ & Lem.\ref{lem:base} & Lem.\ref{lem:base}G & \#$\deg\ge 2^d$
& $KBG$ & \#$\deg\ge 2^{lb}$  & $bsize$
\\
\hline
\hline
\input{tbls/nshat.tex}
\hline
\end{tabular}%
}
\end{table}

\subsection{Analysis: shattered sets and high degree nodes}
\label{sec:analysis}

The running time of our computation is mostly governed by the number of high degree nodes and the number of shattered sets explored, we thus present a detailed analysis of them. Table~\ref{tab:nshat} lists several measures we could perform as follows. By removing the pruning according to Lemma~\ref{lem:base}, our algorithm explores all shattered sets. By setting an initial lower-bound of 0, we first tried to compute all shattered sets and report their number in column ``all''. Note that this computation was not doable within our timeout of 6 hours for six graphs. 
Unsurprisingly, bigger values are observed for larger VC-dimension. 
We then tried to compute all shattered sets included in the set $H'$ of nodes with degree at least $2^d$ where $d$ is the VC-dimension of the graph. Again this computation was out of reach within our time limit for three graphs. We then use $KB$ with $d$ as lower-bound to obtain the number of visited shattered sets in $H'$ according to Lemma~\ref{lem:base} in column ``Lem.\ref{lem:base}''. We obtain similarly column ``Lem.\ref{lem:base}G'' using $KBG$. Column ``KBG'' is obtained using $KBG$ with its heuristic lower-bound (instead of the exact value $d$). We also report the size of $H'$ and $H$ (columns ``\#$\deg\ge 2^d$'' and ``\#$\deg\ge 2^{lb}$'' respectively), and average ball size $bsize$ in the last column.

\begin{figure}[t]
    \centering
    \includegraphics[width=\textwidth]{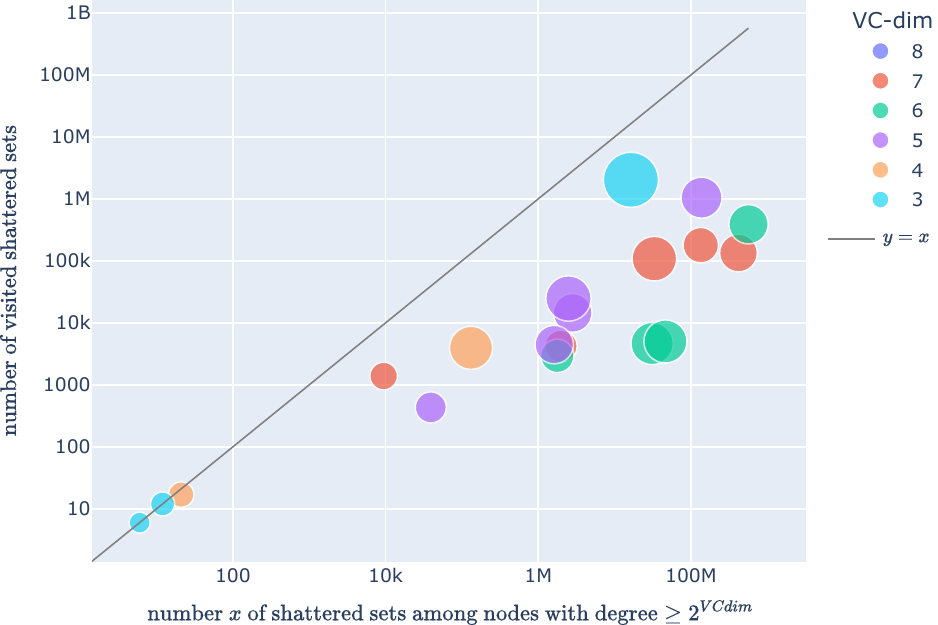} 
    \caption{The number $y$ of visited shattered sets versus the number $x$ of shattered sets in $H'$: each network in the dataset is represented by a disk with coordinates $(x,y)$, whose color indicates the VC-dimension $d$ of the network, while its size is proportional to the logarithm of $|H'|$ ($H'$ denotes the set of nodes with degree $2^d$ at least).}
    \label{fig:nshat-lem1}
\end{figure}

We observe that restricting to $H'$ can reduce the number of shattered sets by a huge factor (e.g. \graph{BIO-SYS-Aff-Cap-RNA-3.5}). 
Note that $H'$ can be empty in graphs with low maximum degree such as grids, explaining the zeros in the table. In general, Lemma~\ref{lem:base} allows to further reduce the number {$y$} of shattered sets to explore as illustrated by
Figure~\ref{fig:nshat-lem1} (column ``$\deg\ge 2^d$'' versus ``Lem.\ref{lem:base}'' of Table~\ref{tab:nshat}).
For networks with more than 10k shattered sets in $H'$, we observe a reduction factor varying from roughly 10 for networks with VC-dimension 3-4, to roughly 100 for networks with VC-dimension 5, and from 1000 to 10k for networks with VC-dimension 6-7 except for \graph{BIO-SYS-Aff-Cap-RNA-3.5} for which the number {$x$} of shattered sets in $H'$ was already quite low. We could not compute the value $x$ for our only network of VC-dimension 8 within our 6 hours limit. Overall, this shows the efficiency of our approach by restricting to high degree nodes and pruning the search by Lemma~\ref{lem:base}.


\subsection{Lower and upper bounds}

As detailed in Appendix~\ref{sec:lb}, the lower bounds we obtain with our heuristic are mostly equal to the true VC-dimension  or just one less. The  upper bounds presented in Lemma~\ref{lem:upper-bounds}, except for the matching number, are analyzed in Appendix~\ref{sec:ub} and are often much larger than the true VC-dimension except for grid-like graphs.

\subsection{Optimizations}

We now compare our reference implementation with other variants of our algorithm obtained by changing the ordering of the vertices ($D^-$, $D^+$, $K$, $R$), using ball restriction ($B$) or not, and using graph restriction ($G$) or not (see Table~\ref{tab:opt}).

\begin{table}[t]
  \caption{Comparing different optimization options of our algorithm with our reference selection ($KBG$). A dash (--) indicates that the timeout of 6 hours was reached.}
  \label{tab:opt}
  \small
\centering
\resizebox{\linewidth}{!}{%
\setlength{\tabcolsep}{4pt}
\begin{tabular}{|l|rrrrrrr|}
\hline
\textbf{Graph} & $D^-BG$ & $D^+BG$ & $\bm{KBG}$ & $RBG$ 
& $KG$ & $KB$  & $K$ 
\\
\hline
\hline
\input{tbls/table3.tex}
\hline
\end{tabular}%
}
\end{table}

Concerning the ordering of the nodes used for scanning shattered sets, we first note that non-decreasing degrees ($D^+$) is almost always faster than non-increasing degrees ($D^-$). Notable exceptions are  \graph{notreDame} and \graph{twitter-combined} where our initial lower bounds 
are $\vcdim(\textrm{notreDame}) -2$ and $\vcdim(\textrm{twitter-combined}) -1$ respectively. 
By observing our traces of execution, we explain it by the fact that the non-increasing order allows to find faster a better lower-bound, which then speeds up the rest of the computation. When the starting lower-bound was indeed exact, non-decreasing degrees is always faster or at least very close to non-increasing degrees. Our intuition is that the number of tentative shattered sets inspected is lower in that case. Indeed, for shattered sets constructed from the first nodes of the ordering, we have to consider all possible remaining nodes (that are in their ball of radius two) and try to construct a tentative shattered set by adding each of them. Putting nodes with lower degree first seem to result in a better balance with respect to the number of tentative shattered sets tested for being shattered.
In that respect, the $k$-core ordering ($K$) seems to work slightly better since it considers the remaining degrees after removing the first nodes rather than degree in the full graph. We note however that a random ordering ($R$) gives overall good results and $RBG$ can even outperform $KBG$ when our lower-bound is not exact, similarly as $D^-BG$ can outperform $D^+BG$. This is in particular the case for \graph{twitter-combined}. Overall there is no clearly better strategy for the ordering and both $k$-core ordering and random ordering seem legitimate choices.

\begin{figure}[t]
    \centering
    \includegraphics[width=.8\textwidth]{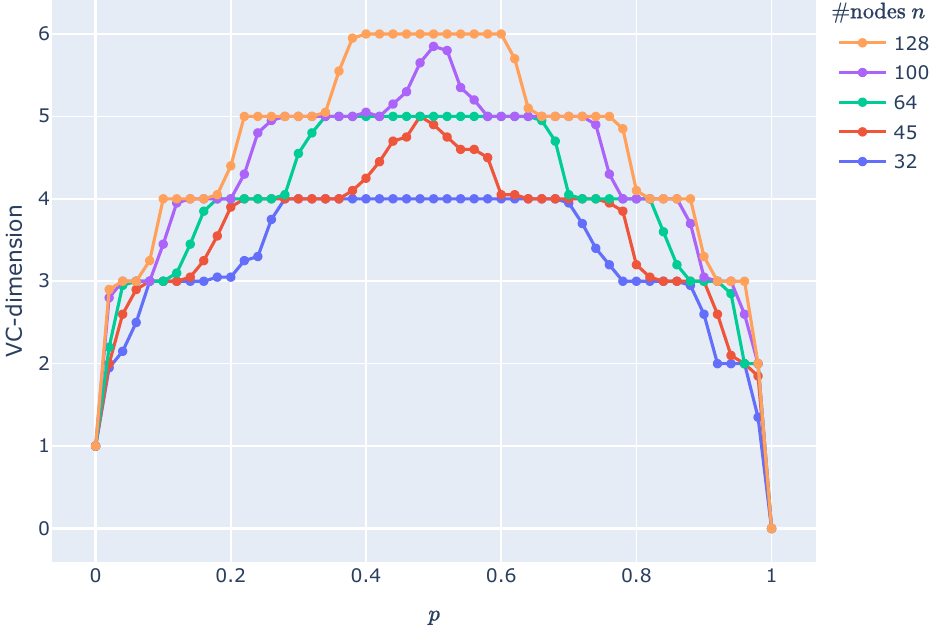}\\
    \includegraphics[width=.8\textwidth]{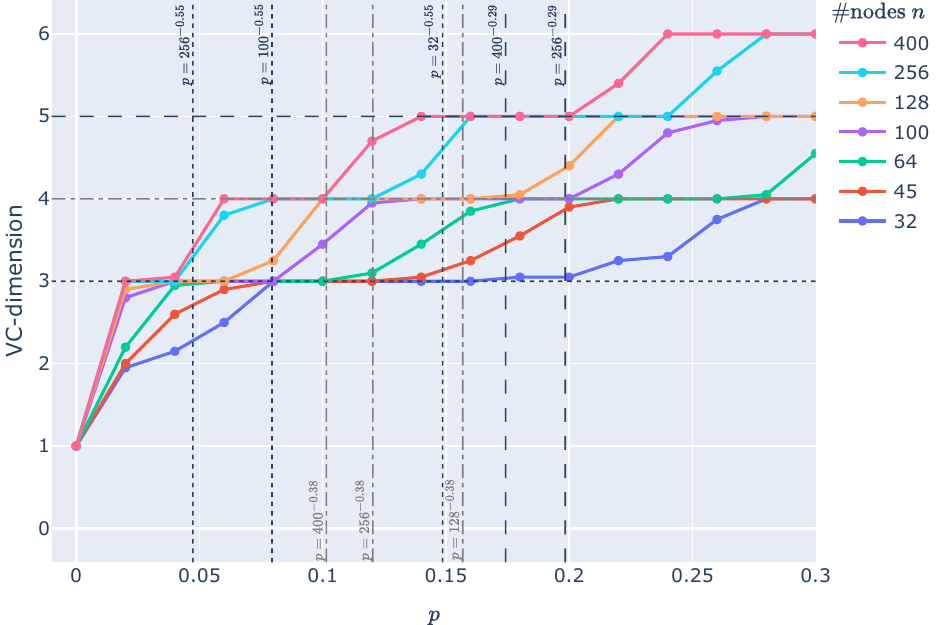}\\
    \caption{Top: the average VC-dimension of $G_{n,p}$ as a function of $p$ for $n=32,45,64,100, 128$. Bottom: a zoom on values $p\in[0,0.3]$ including additional curves for $n=256,400$.}
    \label{fig:erdos-renyi}
\end{figure}

Recall that the ball optimization ($B$) consists in restricting the nodes added to a shattered set $X$ to those that are in the ball of radius two centered at the first node of $X$. It also allows to reduce the number of tentative shattered sets tested and is almost always beneficial. It appears to be mandatory on graphs with very large number of high degree nodes such as \graph{froz} (the number of high degree nodes is analyzed in Section~\ref{sec:analysis}).

The graph reduction optimization ($G$) is not very costly in terms of computation time ($KBG$ is always almost as good as $KB$) and gives a significant improvement on difficult graphs such as \graph{as-skitter} and \graph{y-BerkStan}.

\subsection{VC-dimension of random graphs}

We computed the average VC-dimension of various Erdős-Rényi random graphs $G_{n,p}$ (where each edge appears independently with probability $p$) with up to $n=400$ nodes and compared them to
\cite{anthony1995vapnik} which
proved a threshold of $p=n^{-11/20}=n^{-0.55}$ above which the VC-dimension $d$ of $G_{n,p}$ tends to be greater than $3$, and similarly $p=n^{-21/55}\approx n^{-0.38}$ for $d>4$, and $p=n^{-7/24}\approx n^{-0.29}$ for $d>5$. On the one hand, our results confirm the $p=n^{-0.55}$ and $p=n^{-0.38}$ thresholds which appear to be rather sharp already for $n=100$ or $n=256$. On the other hand, observing the $p=n^{-0.29}$ threshold seems to require size $n$ greater than $400$ (see Figure~\ref{fig:erdos-renyi}).

\begin{figure}[t]
    \centering
    \includegraphics[width=.8\textwidth]{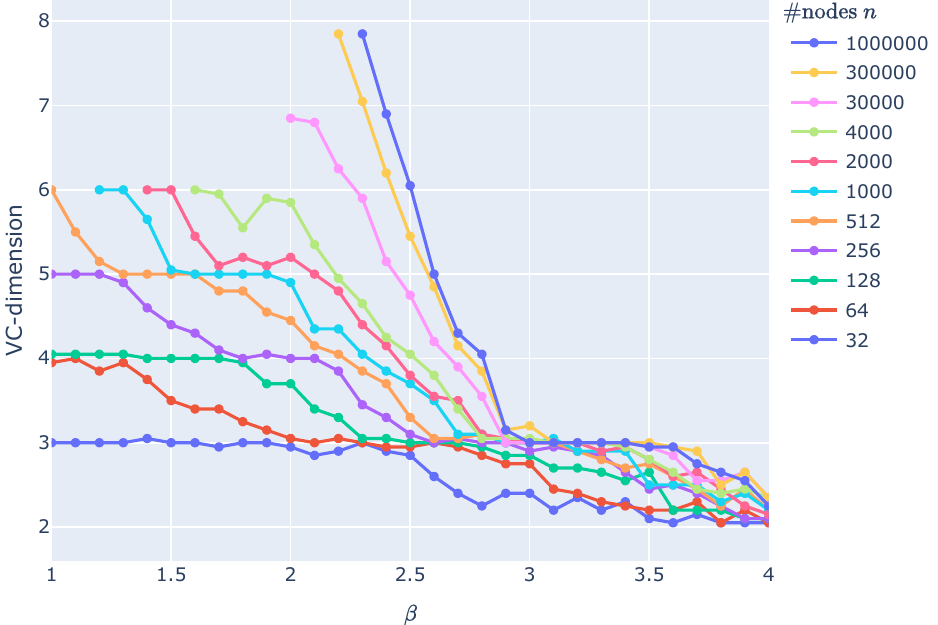} 
    \caption{Average VC-dimension of a power-law random graph with respect to the exponent $\beta$ of the law for various numbers $n$ of nodes. (Curves for large $n$ are truncated for low values of $\beta$.) 
    }
    \label{fig:power-law}
\end{figure}

Figure~\ref{fig:power-law} shows the VC-dimension of various power-law random graphs for  $n\in\{32,64,128,\linebreak 256,512,1\,000,2\,000,4\,000,30\,000,300\,000,1\,000\,000\}$ nodes and exponent $\beta\in [1,4]$.
A graph with $n$ nodes and exponent $\beta$ is obtained by generating a sequence of degrees following a power-law with parameter $\beta$ ($\deg(u)$ gets value $d$ with probability proportional to $1/d^\beta$) and generating a graph according to the configuration model (we generate $\deg(u)$ half edges for each node $u$ and connect half edges according to a random permutation). For each pair $n,\beta$, we have generated 20 random graphs and computed the average VC-dimension of them.
Our curves seem to indicate a threshold value near $\beta=3$ below which the VC-dimension tends to infinity as $n$ grows while the VC-dimension seems to remain constant above (at most three for large $n$).

%% file: tbls/table1.tex
BIO-MV-Physical-3.5 & 21\,866 & 68\,123 & 1\,117 & 5 & 0.11\\
BIO-SYS-Aff-Cap-MS-3.5 & 38\,926 & 299\,855 & 2\,193 & 7 & 4.41\\
BIO-SYS-Aff-Cap-RNA-3.5 & 12\,716 & 40\,541 & 3\,571 & 7 & 0.03\\
dip20170205 & 27\,029 & 73\,762 & 289 & 5 & 0.14\\
\hline
oregon2-010331 & 10\,900 & 31\,180 & 2\,343 & 6 & 0.20\\
CAIDA-as-20130601 & 44\,611 & 151\,434 & 3\,962 & 7 & 11.50\\
DIMES-201204 & 25\,367 & 75\,004 & 3\,781 & 7 & 0.16\\
as-skitter & 1\,696\,415 & 11\,095\,298 & 35\,455 & 8 & 2\,116.90\\
\hline
p2p-Gnutella09 & 8\,114 & 26\,013 & 102 & 5 & 0.01\\
gnutella31 & 62\,586 & 147\,892 & 95 & 4 & 0.29\\
notreDame & 325\,729 & 1\,090\,108 & 10\,721 & 6 & 2.94\\
y-BerkStan & 685\,231 & 6\,649\,470 & 84\,230 & 7 & 11.44\\
\hline
ca-HepPh & 12\,008 & 118\,489 & 491 & 5 & 46.53\\
com-dblp & 317\,080 & 1\,049\,866 & 343 & 5 & 2.30\\
epinions1 & 75\,888 & 405\,740 & 3\,044 & 7 & 58.69\\
facebook-combined & 4\,039 & 88\,234 & 1\,045 & 6 & 8.15\\
twitter-combined & 81\,306 & 1\,342\,296 & 3\,383 & 7 & 456.22\\
\hline
t.CAL & 1\,890\,816 & 2\,315\,222 & 7 & 3 & 4.56\\
t.FLA & 1\,070\,377 & 1\,343\,951 & 8 & 3 & 2.53\\
buddha & 543\,652 & 1\,631\,574 & 17 & 4 & 3.21\\
froz & 754\,197 & 2\,902\,525 & 8 & 3 & 4.70\\
z-alue7065 & 34\,047 & 54\,841 & 4 & 3 & 0.03\\
\hline
grid300-10 & 90\,601 & 162\,535 & 4 & 3 & 0.21\\
xgrid500-10 & 251\,001 & 450\,766 & 4 & 3 & 0.61\\
powerlaw2.5 & 1\,000\,000 & 1\,916\,356 & 9\,447 & 6 & 3.79\\

%% file: tbls/nshat.tex
BIO-MV-Physical-3.5 & 5 & 6\,970\,780   & 1\,618\,308   & 4\,471   & 4\,242   & 582   & 4\,242  & 582  & 118  \\
BIO-SYS-Aff-Cap-MS-3.5 & 7 & --  & 419\,882\,850   & 133\,927   & 133\,546   & 455   & 133\,546  & 455  & 152  \\
BIO-SYS-Aff-Cap-RNA-3.5 & 7 & 1\,552\,654\,200   & 9\,438   & 1\,388   & 1\,246   & 19   & 1\,292  & 152  & 9  \\
dip20170205 & 5 & 9\,050\,178   & 2\,823\,068   & 14\,503   & 14\,227   & 586   & 14\,227  & 586  & 72  \\
\hline
oregon2-010331 & 6 & 18\,763\,142   & 1\,774\,582   & 2\,932   & 2\,610   & 105   & 21\,750  & 258  & 118  \\
CAIDA-as-20130601 & 7 & --  & 134\,557\,021   & 179\,232   & 174\,320   & 168   & 441\,741  & 421  & 217  \\
DIMES-201204 & 7 & 345\,796\,843   & 1\,982\,414   & 4\,206   & 3\,870   & 67   & 3\,870  & 67  & 33  \\
as-skitter & 8 & --  & --  & 6\,671\,019   & 6\,658\,159   & 4\,182   & 6\,658\,159  & 4\,182  & 907  \\
\hline
p2p-Gnutella09 & 5 & 1\,367\,927   & 39\,096   & 433   & 420   & 58   & 420  & 58  & 29  \\
gnutella31 & 4 & 1\,781\,293   & 131\,652   & 3\,958   & 2\,135   & 1\,938   & 19\,417  & 16\,089  & 47  \\
notreDame & 6 & 294\,870\,046   & 46\,413\,353   & 5\,022   & 2\,883   & 2\,063   & 26\,733  & 18\,807  & 507  \\
y-BerkStan & 7 & --  & 33\,179\,823   & 108\,641   & 90\,394   & 1\,540   & 90\,394  & 1\,540  & 395  \\
\hline
ca-HepPh & 5 & 150\,884\,864   & 137\,719\,903   & 1\,046\,942   & 1\,037\,751   & 1\,529   & 1\,037\,751  & 1\,529  & 363  \\
com-dblp & 5 & 17\,874\,027   & 2\,489\,073   & 24\,748   & 14\,329   & 4\,280   & 65\,024  & 20\,566  & 117  \\
epinions1 & 7 & --  & --  & 1\,026\,033   & 1\,025\,863   & 1\,162   & 1\,025\,863  & 1\,162  & 570  \\
facebook-combined & 6 & 722\,273\,307   & 568\,585\,300   & 389\,816   & 389\,471   & 882   & 389\,471  & 882  & 162  \\
twitter-combined & 7 & --  & --  & 1\,709\,219   & 1\,707\,213   & 3\,283   & 2\,562\,770  & 10\,483  & 950  \\
\hline
t.CAL & 3 & 8\,107\,636   & 12   & 12   & 0   & 0   & 0 & 0  & 0  \\
t.FLA & 3 & 4\,805\,712   & 6   & 6   & 0   & 0   & 0 & 0  & 0  \\
buddha & 4 & 7\,958\,071   & 21   & 17   & 0   & 0   & 270\,936  & 119\,239  & 4  \\
froz & 3 & 1\,7402\,814   & 16\,370\,600   & 2\,033\,860   & 2\,026\,837   & 691\,850   & 2\,026\,837  & 691\,850  & 12  \\
z-alue7065 & 3 & 250\,178   & 0   & 0   & 0   & 0   & 0 & 0  & 0  \\
\hline
grid300-10 & 3 & 800\,788   & 0   & 0   & 0   & 0   & 0 & 0  & 0  \\
xgrid500-10 & 3 & 2\,222\,193   & 0   & 0   & 0   & 0   & 0 & 0  & 0  \\
powerlaw2.5 & 6 & 236\,286\,923   & 31\,023\,430   & 4\,653   & 2\,674   & 886   & 2\,674  & 886  & 314  \\

%% file: tbls/table3.tex
BIO-MV-Physical-3.5  & 0.22 & \textbf{0.09} & 0.11 & 0.22 & 0.17 & 0.13 & 0.19\\
BIO-SYS-Aff-Cap-MS-3.5  & 18.09 & \textbf{3.06} & 4.41 & 5.69 & 6.50 & 4.58 & 6.86\\
BIO-SYS-Aff-Cap-RNA-3.5  & \textbf{0.03} & 0.04 & \textbf{0.03} & 0.06 & 0.04 & 0.04 & 0.06\\
dip20170205  & 0.31 & \textbf{0.12} & 0.14 & 0.18 & 0.31 & 0.17 & 0.35\\
\hline
oregon2-010331  & 0.54 & 0.55 & \textbf{0.20} & 0.29 & \textbf{0.20} & 0.25 & 0.25\\
CAIDA-as-20130601  & 11.76 & 17.26 & 11.50 & \textbf{6.76} & 11.85 & 14.10 & 14.06\\
DIMES-201204  & 0.15 & \textbf{0.08} & 0.16 & 0.11 & 0.09 & 0.13 & 0.13\\
as-skitter  & 13\,830.13 & 3\,673.65 & \textbf{2\,116.90} & 4\,022.75 & 8\,892.87 & 2\,358.56 & 10\,102.42\\
\hline
p2p-Gnutella09  & 0.02 & \textbf{0.01} & \textbf{0.01} & 0.03 & 0.02 & 0.02 & \textbf{0.01}\\
gnutella31  & \textbf{0.17} & 0.30 & 0.29 & 0.32 & 15.20 & 0.30 & 16.77\\
notreDame  & \textbf{1.87} & 14.36 & 2.94 & 2.85 & 55.95 & 4.41 & 122.59\\
y-BerkStan  & 11.83 & 11.32 & 11.44 & \textbf{8.33} & 56.86 & 38.47 & 521.76\\
\hline
ca-HepPh  & 138.93 & \textbf{33.54} & 46.53 & 61.00 & 47.03 & 47.82 & 48.18\\
com-dblp  & 1.79 & 2.23 & 2.30 & \textbf{1.71} & 82.82 & 2.68 & 170.05\\
epinions1  & 511.94 & \textbf{31.08} & 58.69 & 144.99 & 58.69 & 63.67 & 63.07\\
facebook-combined  & 21.37 & 10.30 & \textbf{8.15} & 12.09 & 44.56 & 8.24 & 45.04\\
twitter-combined  & 1\,079.24 & 3\,132.27 & 456.22 & \textbf{397.94} & 1\,532.12 & 473.12 & 1\,549.07\\
\hline
t.CAL  & 4.59 & \textbf{4.41} & 4.56 & 4.45 & 4.56 & 4.85 & 4.87\\
t.FLA  & 2.42 & \textbf{2.35} & 2.53 & 2.47 & 2.46 & 2.62 & 2.58\\
buddha  & 2.99 & \textbf{2.92} & 3.21 & 3.01 & 2\,606.40 & 3.22 & 4\,525.09\\
froz  & \textbf{4.39} & 4.40 & 4.70 & 6.44 & -- & 4.45 & --\\
z-alue7065  & \textbf{0.03} & \textbf{0.03} & \textbf{0.03} & 0.04 & \textbf{0.03} & 0.04 & \textbf{0.03}\\
\hline
grid300-10  & 0.22 & 0.23 & \textbf{0.21} & \textbf{0.21} & \textbf{0.21} & 0.22 & \textbf{0.21}\\
xgrid500-10  & 0.61 & \textbf{0.59} & 0.61 & 0.63 & 0.62 & 0.62 & 0.65\\
powerlaw2.5  & 4.06 & \textbf{3.59} & 3.79 & 3.77 & 3.66 & 7.88 & 8.90\\

%% file: 8_conclusion.tex
\section{Conclusion}
\label{sec:conclusion}

A first conclusion is that there is room for improvement with respect to the quick estimation of the VC-dimension, especially concerning good upper bounds. 
Our work can be directly applied in the setups, where the hyperedges are induced by the open neighborhoods or where only the neighborhoods of a subset of vertices are considered as hyperedges.
In future works, besides studying the VC-dimension of power-law random graphs as mentioned above, we will consider extending our algorithm to the practical computation of other VC-dimension parameters on graphs ({\it e.g.}, the distance VC-dimension, see~\cite{chepoi2007covering}), and of the VC-dimension of arbitrary set systems where the ranges are given explicitly. 

%% file: 9_appendix.tex
\section{Proof of Theorem \ref{thm:w1-hard}}
\label{sec:proof-w1-hard}

\begin{proof}
Our starting point is the FPT-reduction of Downey et al.~\cite{downey1993parameterized} who showed that computing the VC-dimension of general set systems is W[1]-hard by the following reduction  from $k$-clique detection.

\begin{theorem}[{Theorem 3 in \cite[Sec. 5]{downey1993parameterized}}]\label{red:generalVC-cliq}
Let $G=(V,E)$ be a graph and $k \leq |V|$ be a parameter. Then $G$ contains a $k$-clique if and only if there is a shattered $k$-subset in the set system ${\cal H}_G = (X,{\cal R})$, where  $X = V \times \{1,2,\ldots,k\}$ and ${\cal R} = {\cal R}_0 \cup {\cal R}_1 \cup {\cal R}_2 \cup {\cal R}_3$ with
\begin{itemize}
    \item[] ${\cal R}_0 = \{\emptyset\}$;
    \item[] ${\cal R}_1 = \{ \{(v,i)\} \mid v \in V, \ 1 \leq i \leq k \}$;
    \item[] ${\cal R}_2 = \{ \{(u,i),(v,j)\} \mid uv \in E, \ 1 \leq i,j \leq k \}$;
    \item[] ${\cal R}_3 = \{ \{ (v,i) \mid v \in V, \ i \in S \} \mid S \subseteq \{1,2,\ldots,k\} \ \text{and} \ |S| \geq 3 \}$.
\end{itemize}
\end{theorem}

\noindent
Let $\mathcal{H}_G = (X,{\cal R})$ be defined as in Theorem~\ref{red:generalVC-cliq}. 
Note that the cumulative cardinalities of all sets in ${\cal R}$ sum up to $\bigO(kn+k^2m+k2^kn) = \bigO(k2^kn^2)$.\\
We construct a graph $H_G$ from $\mathcal{H}_G$ as follows: 
\begin{itemize}
    \item $V(H_G) = X \cup {\cal R}$;
    \item for every $x \in X$ and  $R \in {\cal R}$ such that $x \in R$, add the edge $\{x,R\}$ to $E(H_G)$;
    \item if $k \geq 3$, for every distinct $x,y \in X$ add the edge $\{x,y\}$ to $E(H_G)$. 
\end{itemize}
Note that the cost of computing $H_G$ is that of computing $\mathcal{H}_G$, with an additional cost in $\bigO(|X|^2) = \bigO(k^2n^2)$ if $k \ge 3$.
It is easy to see that if $Y \subseteq X$ is shattered in $\mathcal{H}_G$, then it is also shattered by the closed neighborhoods in $H_G$.
Thus, \Cref{red:generalVC-cliq} implies that if $G$ has a $k$-clique, then the VC-dimension of $\mathcal{H}_G$, and so the VC-dimension of $H_G$, is at least $k$.

It remains to show that if $G$ has no $k$-clique (or equivalently, the VC-dimension of $\mathcal{H}_G$ is strictly less than $k$), then the VC-dimension of $H_G$ is at most $k-1$.

If $k=0$, then $G$ must be an empty graph. Then, ${\cal R} = {\cal R}_0$. It implies that $H_G$ is reduced to one vertex, and so, its VC-dimension equals $0$.

If $k=1$, then $G$ is a stable set which implies $\mathcal R_2 = \emptyset$, and we also have $\mathcal R_3 = \emptyset$. In particular, ${\cal R} = {\cal R}_0 \cup {\cal R}_1$. By definition, $H_G$ contains a perfect matching between $X$ and $\R_1$, plus an isolated vertex for ${\cal R}_0$. Therefore, the VC-dimension of $H_G$ equals $1$.

Assume now that $\omega(G) \le k-1$ with $k \geq 3 $ and suppose for the sake of contradiction that there exists a $k$-element subset $Y\subseteq V(H_G)$ which is shattered by closed neighborhoods in $H_G$.
By definition, $H_G$ is a split graph with clique $X$ and independent set ${\cal R}$.

\begin{claim*}
We either have  $Y \subseteq X$ or $Y \subseteq {\cal R}$. 
\end{claim*}
\begin{proof}
The proof is based on observations similar to the ones in~\cite[Proof of Lemma $11$]{ducoffe2022diameter}.
Assume that $Y$ intersects both $X$ and $\R$.
Since $|Y| \geq 3$, we can take distinct elements $x,y,z \in Y$ with $x \in X$, $z \in \R$ and $y$ belonging to either $X$ or $\R$. 
Suppose first that $y \notin N[z]$.
Let $v \in V(H_G)$ be such that $N[v] \cap \{x,y,z\} = \{y,z\}$.
Since $y \notin N[z]$ and ${\cal R}$ is a stable set, necessarily $v \in X$.
But then, $x \in N[v]$ because $X$ is a clique, thus contradicting that $N[v] \cap \{x,y,z\} = \{y,z\}$.
As a result, we must have $y \in N[z]$.
Let $u \in V(H_G)$ be such that $N[u] \cap \{x,y,z\} = \{x,z\}$.
Since $u,y \in N[z]$ and $N[z]$ is a clique, necessarily $y \in N[u]$.
Again, the latter contradicts our assumption that $N[u] \cap \{x,y,z\} = \{x,z\}$.
\end{proof}

First we suppose that $Y \subseteq X$.
Since we assumed that the VC-dimension of $\mathcal{H}_G$ is strictly less than $k$, $Y$ is not shattered in $\mathcal{H}_G$.
Thus  
\begin{equation}\label{eq:neigh-constraint}
    \text{there exists $Z \subseteq Y$ that is not a trace of any range in $\R$.}
\end{equation} 
Since $Y$ is shattered by $H_G$, there is a vertex $v_Z \in V(H_G)$ such that $N[v_Z] \cap Y = Z$. The definition of $H_G$ and \eqref{eq:neigh-constraint} imply that
there is no $v \in \R$ such that its neighborhood in $H_G$ satisfies $N[v] \cap Y = Z$ and so we get that $v_Z \in X$. Since $H_G[X]$ is a clique, we conclude that $Z = Y$. However, there exists a range in ${\cal R}_3$ which contains the entire set $Y$ (namely, the one corresponding to the full set $S=\{1,2\ldots,k\}$), a contradiction with \eqref{eq:neigh-constraint}.


Finally, we consider the case of $Y \subseteq {\cal R}$.
Let $y \in Y$ be any vertex. Since $Y$ is shattered, each of the $2^{k-1}$ subsets of $Y$ that contain $y$ is either the trace of $N[y]$ or the trace of $N[x]$ for some neighbor $x$ of $y$. Therefore, $|N[y]| \ge 2^{k-1} \geq 4$, and so the range in $\mathcal H_G$ corresponding to $y$ has size at least $3$, in particular, $y \in {\cal R}_3$.
For any $Y \subseteq \R_3$, the closed neighborhoods of the vertices of $H_G$ can have the following traces on $Y$:
\begin{itemize}
    \item if $v \in \R$, then $N[v] \cap Y$ is equal to $\{v\}$ if $ v \in Y$ and $\emptyset$ otherwise;
    \item if $v = (x,i) \in X$, then $N[v] \cap Y$ contains those vertices of $Y$ that correspond to index sets $S$ with $i \in S$ (see the definition of $\R_3$).
\end{itemize}
That is, the neighborhoods of vertices in $\R$ can only induce the empty set and the $k$ singleton traces on $Y$,  and for any $x,y \in V(G)$ and $i \in \{1,2\ldots,k\}$, we have $N[(x,i)] \cap Y = N[(y,i)] \cap Y$.
This implies that the number of vertices in $X$ that have pairwise different neighborhoods in $Y$ is at most $k$. On the other hand, since $Y$ is shattered, we need to obtain each of the $2^{k}$ subsets of $Y$ as a trace, which implies that $2^{k} \leq k + k +1$, and thus $k \leq 2$, a contradiction.
\end{proof}

\section{Simple bounds}
\label{sec:more-simple-bounds}

\begin{lemma}\label{lem:degeneracy}
    A $k$-degenerate graph has VC-dimension at most $k+1$,
    and this bound is sharp.
\end{lemma}

\begin{proof}
    Let $G = (V, E)$ be a $k$-degenerate graph and consider a shattered set $X\subseteq V$. Let $Z = X \cup \{v_Y \mid Y \subseteq X\}$, where $v_Y$ denotes an arbitrary vertex such that $N[v_Y] \cap X = Y$. We consider the induced subgraph $G[Z]$ and we iteratively remove all vertices of $Z \setminus X$ with at most $k$ neighbours. Let $Z'$ be the set of remaining vertices. Note that $X \subseteq Z'$. Since $G[Z']$ is also $k$-degenerate, it has some vertex $x$ with at most $k$ neighbours. Since we iteratively removed all vertices with at most $k$ neighbours in $Z \setminus X$, then necessarily $x \in X$. Furthermore, all vertices of $Z \setminus Z'$ must be of the form $v_Y$ for some $|Y| \leq k$. 
    As a result, the number of neighbours of $x$ in $G[Z]$ is no more than $k + \sum_{i=0}^{k-1}\binom{|X|-1}{i}$. 
    However, since $X$ is shattered, and there are $2^{|X|-1}$ subsets of $X$ containing vertex $x$, we must have $|N[x]\cap Z| \ge 2^{|X|-1}$.
    In particular, the number of neighbours of $x$ in $G[Z]$ must be at least $2^{|X|-1}-1$.
    Suppose by contradiction that $k < |X|-1$. 
    Then, 
    \begin{align*}
        k + \sum_{i=0}^{k-1}\binom{|X|-1}{i} &= k + 2^{|X|-1} - \sum_{i=k}^{|X|-1} \binom{|X|-1}{i}  \\
        &\le k + 2^{|X|-1} - \sum_{i=|X|-2}^{|X|-1} \binom{|X|-1}{i} \\
        &= k + 2^{|X|-1} - \binom{|X|-1}{|X|-2} - \binom{|X|-1}{|X|-1} \\
        &= k + 2^{|X|-1} - |X| \\
        &< 2^{|X|-1} -1
    \end{align*}
    A contradiction.
    Hence, $|X| \leq k+1$.
    This is sharp for $k=1$ because trees are $1$-degenerate, and there exist $2$-shattered subsets in trees ({\it e.g.}, any two leaves in a star with at least three leaves).
\end{proof}

\begin{table}[t]
  \caption{VC-dimension lower bounds computed with $maxvisits=16,32,64,128,256$ (bold values indicate that the bound matches the exact value), and the corresponding execution time, the ``read'' column corresponds to the time for reading the graph.}
  \label{tab:lb}
  \small
\centering
\resizebox{.9\linewidth}{!}{%
\setlength{\tabcolsep}{4pt}
\begin{tabular}{|l|cccccc|lrrrrr|}
\hline
& & \multicolumn{5}{c|}{lower-bound}
& \multicolumn{6}{c|}{time (s)}\\
\hline
\textbf{Graph} 
& VC-dim
& 16 & 32 & 64 & 128 & 256
& read
& 16 & 32 & 64 & 128 & 256
\\
\hline
\hline
\input{tbls/table2.tex}
\hline
\end{tabular}%
}
\end{table}

Finally, we show upper bounds on the VC-dimension in terms of sizes of maximum and maximal matchings.
\begin{lemma}\label{lem:matching}
    Let $G=(V, E)$ be a non-empty graph and $M$ be  a maximal matching of $G$. Then the VC-dimension of $G$ is at most $2|M|$. Moreover, if $\nu(G)$ is the size of a maximum matching in $G$, then we have $\vcdim(G) \leq \nu(G)+1$.
\end{lemma}
\begin{proof}
Since $G$ is non-empty, any maximal matching has at least one edge and thus the statements trivially hold if $\vcdim(G) \leq 2$. 
    Let $X \subseteq V$ be a shattered set of size $\vcdim(G)$. If $M$ covers every vertex of $X$, then we have $|M| \geq \frac{1}{2}\cdot \vcdim(G)$. Assume that there exists $x \in X$ which is not covered by $M$. Since $M$ is maximal, each of the at least $2^{\vcdim(G)-1}-1$ neighbors 
    of $x$ need to be covered by $M$, which implies $|M|\geq \frac 1 2 \cdot \left(2^{\vcdim(G)-1} -1\right) \geq \frac{1}{2} \vcdim(G)$ for any graph with $\vcdim(G) \geq 3$.

    To show that $\vcdim(G) \leq \nu(G)+1$, it  is sufficient to construct a matching where $|X|-1$ vertices of $X$ are matched to a vertex outside of $X$. Since $X$ is shattered, for any $x \in X$, there exists a vertex $v_x \in V(G)$ such that $N[v_x] \cap X = \{x\}$. Observe that  we have either  $v_x \notin X$ or $v_x = x$  and the second option is only possible if $x$ has no neighbors in $X$. 
    We build a matching $M$ of size $|X|-1$ as follows. First we add all edges $\{x,v_x\}$ to $M$ where $x \in X$ is such that it has at least one neighbor in $X$. After this, if there is a remaining set $Y \subseteq X$ which is not yet covered by $M$, then $Y$ has to be a stable set. Since $X$ is shattered, there is a vertex $v_{Y}$ such that  $v_{Y}\cap X = Y$. As $Y$ is a stable set,  $v_{Y} \notin Y$ if $|Y|\geq 2$. Thus, we can cover each element of $Y$ (except maybe one) by taking any $y \in Y$, adding $\{v_{Y}, y\}$ to $M$, and recursing on $Y' = Y \setminus \{y\}$. In the end, we get a matching of size $|M| \geq |X\setminus Y| + |Y| - 1 = \vcdim(G) -1$.
\end{proof}

\noindent
Note that stars with at least three leaves have matching number one and VC-dimension two, thus the bound of Lemma~\ref{lem:matching} is sharp.

\begin{table}[b]
  \caption{Some upper-bounds of Lemma~\ref{lem:upper-bounds} compared to lower-bounds and true VC-dimension.}
  \label{tab:bounds}
  \small
\centering
\resizebox{.69\linewidth}{!}{%
\setlength{\tabcolsep}{4pt}
\begin{tabular}{|l|rrrrr|}
\input{tbls/bounds.tex}
\end{tabular}%
}
\end{table}

\section{Experiments}

\subsection{Lower-bound computation}
\label{sec:lb}

Table~\ref{tab:lb} gives the lower-bounds obtained by our lower-bound heuristic for various values of $maxvisits$ while $KBG$ uses $maxvisits=64$. It also provides the corresponding running times for reading the graph and computing the lower-bound. As a reference, column ``read'' indicates the time spent for just reading the graph.

We observe a tradeof where increasing $maxvisits$ provides generally a better lower-bound at the cost of a longer running time. Note the exception of \graph{twitter-combined}
for which the best bound is obtained only for $maxvisits=32$.
We also note that the running time stays within a factor 5 of the time taken for reading the graph, even for $maxvisits=256$. The exception of \graph{twitter-combined} let us think that there is room for improvement of the tuning of our heuristic. For example, the choice of $maxvisits/2$ for limiting the for loop of \textsc{ExploreShattered} was not intensively explored.
However, it already provides lower-bounds which are often exact or one less than the true VC-dimension. This is indeed the case for $maxivists=128$, and almost the case for $maxvisits=64$ where \graph{notreDame} is the only exception with a lower-bound which is two less than the VC-dimension.


\subsection{Upper bounds}
\label{sec:ub}

Table~\ref{tab:bounds} lists the upper-bounds we can quickly compute on our dataset. The degree upper-bound $\floor{\log \Delta}+1$ where $\Delta$ is the maximum degree appears to always be the best one. However, it can be as large as twice the true value and gives a poor confidence bound compared to what we obtained for lower-bounds. The node upper-bound $\floor{\log n}$ where $n$ is the number of nodes is almost always greater. The only graph where it matches the degree upper-bound is \graph{facebook-combined} which has high maximum degree $\Delta$ compared to its number $n$ of nodes as it satisfies $\Delta > n/4$ (see Table~\ref{tab:ref}). Finally, the degeneracy upper-bound appears to be good on graphs with low degree, that is road networks and grid like graphs, while it can be very high for graphs with many high degree nodes. Note that it can be quite high even for graphs with relatively low VC-dimension such as \graph{ca-HepPh}.






%% file: tbls/table2.tex
BIO-MV-Physical-3.5 & \textbf{5} & \textbf{5} & \textbf{5} & \textbf{5} & \textbf{5} & \textbf{5} & 0.06 & 0.05 & 0.05 & 0.05 & 0.05 & 0.10 \\
BIO-SYS-Aff-Cap-MS-3.5 & \textbf{7} & 6 & 6 & \textbf{7} & \textbf{7} & \textbf{7} & 0.15 & 0.18 & 0.18 & 0.18 & 0.20 & 0.23 \\
BIO-SYS-Aff-Cap-RNA-3.5 & \textbf{7} & 6 & 6 & 6 & 6 & \textbf{7} & 0.02 & 0.03 & 0.02 & 0.03 & 0.03 & 0.04 \\
dip20170205 & \textbf{5} & \textbf{5} & \textbf{5} & \textbf{5} & \textbf{5} & \textbf{5} & 0.04 & 0.05 & 0.05 & 0.05 & 0.06 & 0.07 \\
\hline
oregon2-010331 & \textbf{6} & 5 & 5 & 5 & 5 & 5 & 0.01 & 0.02 & 0.02 & 0.02 & 0.02 & 0.04 \\
CAIDA-as-20130601 & \textbf{7} & 6 & 6 & 6 & 6 & \textbf{7} & 0.07 & 0.09 & 0.09 & 0.09 & 0.11 & 0.13 \\
DIMES-201204 & \textbf{7} & 6 & \textbf{7} & \textbf{7} & \textbf{7} & \textbf{7} & 0.04 & 0.05 & 0.05 & 0.05 & 0.05 & 0.06 \\
as-skitter & \textbf{8} & 7 & 7 & \textbf{8} & \textbf{8} & \textbf{8} & 6.11 & 7.04 & 7.42 & 7.98 & 8.24 & 9.92 \\
\hline
p2p-Gnutella09 & \textbf{5} & \textbf{5} & \textbf{5} & \textbf{5} & \textbf{5} & \textbf{5} & 0.01 & 0.02 & 0.02 & 0.02 & 0.02 & 0.03 \\
gnutella31 & \textbf{4} & 3 & 3 & 3 & 3 & 3 & 0.12 & 0.11 & 0.11 & 0.13 & 0.18 & 0.26 \\
notreDame & \textbf{6} & 3 & 3 & 4 & 5 & \textbf{6} & 0.74 & 1.05 & 1.16 & 1.12 & 1.19 & 1.20 \\
y-BerkStan & \textbf{7} & 5 & \textbf{7} & \textbf{7} & \textbf{7} & \textbf{7} & 3.49 & 4.87 & 4.88 & 5.08 & 5.39 & 5.88 \\
\hline
ca-HepPh & \textbf{5} & \textbf{5} & \textbf{5} & \textbf{5} & \textbf{5} & \textbf{5} & 0.06 & 0.07 & 0.10 & 0.08 & 0.10 & 0.14 \\
com-dblp & \textbf{5} & 4 & 4 & 4 & 4 & 0 & 0.58 & 0.71 & 0.78 & 0.81 & 0.95 & 0.00 \\
epinions1 & \textbf{7} & 6 & \textbf{7} & \textbf{7} & \textbf{7} & \textbf{7} & 0.25 & 0.33 & 0.35 & 0.36 & 0.40 & 0.46 \\
facebook-combined & \textbf{6} & 5 & \textbf{6} & \textbf{6} & \textbf{6} & \textbf{6} & 0.04 & 0.04 & 0.05 & 0.06 & 0.07 & 0.10 \\
twitter-combined & \textbf{7} & 6 & \textbf{7} & 6 & 6 & 6 & 0.64 & 0.70 & 0.70 & 0.82 & 1.03 & 1.45 \\
\hline
t.CAL & \textbf{3} & \textbf{3} & \textbf{3} & \textbf{3} & \textbf{3} & \textbf{3} & 3.30 & 4.61 & 4.53 & 4.52 & 4.53 & 4.60 \\
t.FLA & \textbf{3} & \textbf{3} & \textbf{3} & \textbf{3} & \textbf{3} & \textbf{3} & 1.83 & 2.45 & 2.50 & 2.45 & 2.44 & 2.43 \\
buddha & \textbf{4} & 3 & 3 & 3 & 3 & 3 & 2.13 & 2.73 & 2.63 & 3.27 & 3.16 & 4.13 \\
froz & \textbf{3} & \textbf{3} & \textbf{3} & \textbf{3} & \textbf{3} & \textbf{3} & 1.43 & 1.98 & 2.54 & 3.02 & 4.46 & 6.91 \\
z-alue7065 & \textbf{3} & \textbf{3} & \textbf{3} & \textbf{3} & \textbf{3} & \textbf{3} & 0.03 & 0.04 & 0.04 & 0.03 & 0.04 & 0.03 \\
\hline
grid300-10 & \textbf{3} & \textbf{3} & \textbf{3} & \textbf{3} & \textbf{3} & \textbf{3} & 0.16 & 0.21 & 0.23 & 0.22 & 0.21 & 0.20 \\
xgrid500-10 & \textbf{3} & \textbf{3} & \textbf{3} & \textbf{3} & \textbf{3} & \textbf{3} & 0.50 & 0.61 & 0.83 & 0.62 & 0.62 & 0.64 \\
powerlaw2.5 & \textbf{6} & \textbf{6} & \textbf{6} & \textbf{6} & \textbf{6} & \textbf{6} & 1.38 & 3.33 & 3.50 & 3.36 & 3.49 & 3.70 \\

%% file: tbls/bounds.tex
\hline
\textbf{graph} & \textbf{lb64} & \textbf{VCdim} & \textbf{$\floor{\log\Delta}+1$} & \textbf{$\floor{\log{n}}$} & \textbf{degen+1}\\
\hline
BIO-MV-Physical-3.5 & 5 & \textbf{5} & 11 & 14 & 37\\
BIO-SYS-Aff-Cap-MS-3.5 & 7 & \textbf{7} & 12 & 15 & 52\\
BIO-SYS-Aff-Cap-RNA-3.5 & 6 & \textbf{7} & 12 & 13 & 55\\
dip20170205 & 5 & \textbf{5} & 9 & 14 & 22\\
\hline
oregon2-010331 & 5 & \textbf{6} & 12 & 13 & 32\\
CAIDA-as-20130601 & 6 & \textbf{7} & 12 & 15 & 69\\
DIMES-201204 & 7 & \textbf{7} & 12 & 14 & 35\\
as-skitter & 8 & \textbf{8} & 16 & 20 & 112\\
\hline
p2p-Gnutella09 & 5 & \textbf{5} & 7 & 12 & 11\\
gnutella31 & 3 & \textbf{4} & 7 & 15 & 7\\
notreDame & 4 & \textbf{6} & 14 & 18 & 156\\
y-BerkStan & 7 & \textbf{7} & 17 & 19 & 202\\
\hline
ca-HepPh & 5 & \textbf{5} & 9 & 13 & 239\\
com-dblp & 4 & \textbf{5} & 9 & 18 & 114\\
epinions1 & 7 & \textbf{7} & 12 & 16 & 68\\
facebook-combined & 6 & \textbf{6} & 11 & 11 & 116\\
twitter-combined & 6 & \textbf{7} & 12 & 16 & 97\\
\hline
t.CAL & 3 & \textbf{3} & 4 & 20 & 4\\
t.FLA & 3 & \textbf{3} & 4 & 20 & 4\\
buddha & 3 & \textbf{4} & 5 & 19 & 6\\
froz & 3 & \textbf{3} & 4 & 19 & 5\\
z-alue7065 & 3 & \textbf{3} & 3 & 15 & 3\\
\hline
grid300-10 & 3 & \textbf{3} & 3 & 16 & 3\\
xgrid500-10 & 3 & \textbf{3} & 3 & 17 & 3\\
powerlaw2.5 & 6 & \textbf{6} & 14 & 19 & 19\\
\hline

%% file: 0_main_hal.bbl
\begin{thebibliography}{10}

\bibitem{anthony1995vapnik}
Martin Anthony, Graham Brightwell, and Colin Cooper.
\newblock The {V}apnik-{C}hervonenkis dimension of a random graph.
\newblock {\em Discrete Mathematics}, 138(1-3):43--56, 1995.
\newblock \href {https://doi.org/10.1016/0012-365X(94)00187-N}
  {\path{doi:10.1016/0012-365X(94)00187-N}}.

\bibitem{blumer1989learnability}
Anselm Blumer, Andrzej Ehrenfeucht, David Haussler, and Manfred~K. Warmuth.
\newblock Learnability and the {V}apnik-{C}hervonenkis dimension.
\newblock {\em Journal of the ACM (JACM)}, 36(4):929--965, 1989.
\newblock \href {https://doi.org/10.1145/76359.76371}
  {\path{doi:10.1145/76359.76371}}.

\bibitem{bonamy2021eptas}
Marthe Bonamy, {\'E}douard Bonnet, Nicolas Bousquet, Pierre Charbit, Panos
  Giannopoulos, Eun~Jung Kim, Pawe{\l} Rz{{a}}{\.z}ewski, Florian Sikora, and
  St{\'e}phan Thomass{\'e}.
\newblock {EPTAS} and subexponential algorithm for maximum clique on disk and
  unit ball graphs.
\newblock {\em Journal of the ACM (JACM)}, 68(2):1--38, 2021.
\newblock \href {https://doi.org/10.1145/3433160} {\path{doi:10.1145/3433160}}.

\bibitem{bousquet2015identifying}
Nicolas Bousquet, Aur{\'e}lie Lagoutte, Zhentao Li, Aline Parreau, and
  St{\'e}phan Thomass{\'e}.
\newblock Identifying codes in hereditary classes of graphs and {VC}-dimension.
\newblock {\em SIAM Journal on Discrete Mathematics}, 29(4):2047--2064, 2015.
\newblock \href {https://doi.org/10.1137/14097879}
  {\path{doi:10.1137/14097879}}.

\bibitem{chalopin2023sample}
J{\'e}r{\'e}mie Chalopin, Victor Chepoi, Fionn Mc~Inerney, S{\'e}bastien Ratel,
  and Yann Vax{\`e}s.
\newblock Sample compression schemes for balls in graphs.
\newblock {\em SIAM Journal on Discrete Mathematics}, 37(4):2585--2616, 2023.

\bibitem{chalopin2022unlabeled}
J{\'e}r{\'e}mie Chalopin, Victor Chepoi, Shay Moran, and Manfred~K. Warmuth.
\newblock Unlabeled sample compression schemes and corner peelings for ample
  and maximum classes.
\newblock {\em Journal of Computer and System Sciences}, 127:1--28, 2022.
\newblock \href {https://doi.org/10.1016/j.jcss.2022.01.003}
  {\path{doi:10.1016/j.jcss.2022.01.003}}.

\bibitem{chazelle1989quasi}
Bernard Chazelle and Emo Welzl.
\newblock Quasi-optimal range searching in spaces of finite {VC}-dimension.
\newblock {\em Discrete \& Computational Geometry}, 4:467--489, 1989.
\newblock \href {https://doi.org/10.1007/BF02187743}
  {\path{doi:10.1007/BF02187743}}.

\bibitem{chepoi2007covering}
Victor Chepoi, Bertrand Estellon, and Yann Vaxes.
\newblock Covering planar graphs with a fixed number of balls.
\newblock {\em Discrete \& Computational Geometry}, 37:237--244, 2007.
\newblock \href {https://doi.org/10.1007/s00454-006-1260-0}
  {\path{doi:10.1007/s00454-006-1260-0}}.

\bibitem{chepoi2020density}
Victor Chepoi, Arnaud Labourel, and S{\'e}bastien Ratel.
\newblock On density of subgraphs of {C}artesian products.
\newblock {\em Journal of Graph Theory}, 93(1):64--87, 2020.
\newblock \href {https://doi.org/10.1002/jgt.22469}
  {\path{doi:10.1002/jgt.22469}}.

\bibitem{csikos2022optimal}
M{\'o}nika Csik{\'o}s and Nabil~H. Mustafa.
\newblock Optimal approximations made easy.
\newblock {\em Information Processing Letters}, 176:106250, 2022.
\newblock \href {https://doi.org/10.1016/j.ipl.2022.106250}
  {\path{doi:10.1016/j.ipl.2022.106250}}.

\bibitem{demaine2019structural}
Erik~D. Demaine, Felix Reidl, Peter Rossmanith, Fernando~S{\'a}nchez Villaamil,
  Somnath Sikdar, and Blair~D. Sullivan.
\newblock Structural sparsity of complex networks: Bounded expansion in random
  models and real-world graphs.
\newblock {\em Journal of Computer and System Sciences}, 105:199--241, 2019.
\newblock \href {https://doi.org/10.1016/j.jcss.2019.05.004}
  {\path{doi:10.1016/j.jcss.2019.05.004}}.

\bibitem{despres2017vapnikchervonenkis}
Christian J.~J. Despres.
\newblock The {V}apnik-{C}hervonenkis dimension of cubes in $\mathbb{R}^d$,
  2017.
\newblock \href {https://arxiv.org/abs/1412.6612} {\path{arXiv:1412.6612}}.

\bibitem{downey1993parameterized}
Rodney~G. Downey, Patricia~A. Evans, and Michael~R. Fellows.
\newblock Parameterized learning complexity.
\newblock In {\em Proceedings of the sixth annual conference on Computational
  learning theory}, pages 51--57, 1993.

\bibitem{ducoffe2021computing}
Guillaume Ducoffe.
\newblock On computing the average distance for some chordal-like graphs.
\newblock In {\em 46th International Symposium on Mathematical Foundations of
  Computer Science (MFCS)}, 2021.

\bibitem{ducoffe2022diameter}
Guillaume Ducoffe.
\newblock The diameter of {AT}-free graphs.
\newblock {\em Journal of Graph Theory}, 99(4):594--614, 2022.
\newblock \href {https://doi.org/10.1002/jgt.22754}
  {\path{doi:10.1002/jgt.22754}}.

\bibitem{ducoffe2022diameterb}
Guillaume Ducoffe, Michel Habib, and Laurent Viennot.
\newblock Diameter, eccentricities and distance oracle computations on
  {H}-minor free graphs and graphs of bounded (distance)
  {V}apnik-{C}hervonenkis dimension.
\newblock {\em SIAM Journal on Computing}, 51(5):1506--1534, 2022.
\newblock \href {https://doi.org/10.1137/20M136551}
  {\path{doi:10.1137/20M136551}}.

\bibitem{floyd1995sample}
Sally Floyd and Manfred Warmuth.
\newblock Sample compression, learnability, and the {V}apnik-{C}hervonenkis
  dimension.
\newblock {\em Machine learning}, 21:269--304, 1995.
\newblock \href {https://doi.org/10.1023/A:1022660318680}
  {\path{doi:10.1023/A:1022660318680}}.

\bibitem{FPS21}
Jacob Fox, János Pach, and Andrew Suk.
\newblock Bounded {VC}-dimension implies the {S}chur-{Erd\H{o}s} conjecture.
\newblock {\em Combinatorica}, 41(6):803--813, 2021.
\newblock \href {https://doi.org/10.1007/s00493-021-4530-9}
  {\path{doi:10.1007/s00493-021-4530-9}}.

\bibitem{HMPV00}
Michel Habib, Ross McConnell, Christophe Paul, and Laurent Viennot.
\newblock Lex-{BFS} and partition refinement, with applications to transitive
  orientation, interval graph recognition and consecutive ones testing.
\newblock {\em Theoretical Computer Science}, 234(1-2):59--84, 2000.
\newblock \href {https://doi.org/10.1016/S0304-3975(97)00241-7}
  {\path{doi:10.1016/S0304-3975(97)00241-7}}.

\bibitem{haussler1986epsilon}
David Haussler and Emo Welzl.
\newblock Epsilon-nets and simplex range queries.
\newblock In {\em Proceedings of the second annual symposium on Computational
  geometry}, pages 61--71, 1986.

\bibitem{holden1995practical}
Sean~B. Holden and Mahesan Niranjan.
\newblock On the practical applicability of {VC}-dimension bounds.
\newblock {\em Neural Computation}, 7(6):1265--1288, 1995.
\newblock \href {https://doi.org/10.1162/neco.1995.7.6.1265}
  {\path{doi:10.1162/neco.1995.7.6.1265}}.

\bibitem{kranakis1997vc}
Evangelos Kranakis, Danny Krizanc, Berthold Ruf, Jorge Urrutia, and Gerhard
  Woeginger.
\newblock The {VC}-dimension of set systems defined by graphs.
\newblock {\em Discrete Applied Mathematics}, 77(3):237--257, 1997.
\newblock \href {https://doi.org/10.1016/S0166-218X(96)00137-0}
  {\path{doi:10.1016/S0166-218X(96)00137-0}}.

\bibitem{P2P}
Jure Leskovec, Jon Kleinberg, and Christos Faloutsos.
\newblock Graph evolution: Densification and shrinking diameters.
\newblock {\em ACM transactions on Knowledge Discovery from Data - {TKDD}},
  1(1):2--42, 2007.
\newblock \href {https://doi.org/10.1145/1217299.1217301}
  {\path{doi:10.1145/1217299.1217301}}.

\bibitem{LiLS-sample-complexity-learning-S01}
Yi~Li, Philip~M. Long, and Aravind Srinivasan.
\newblock Improved {B}ounds on the {S}ample {C}omplexity of {L}earning.
\newblock {\em Journal of Computer and System Sciences}, 62(3):516--527, 2001.
\newblock \href {https://doi.org/10.1006/jcss.2000.1741}
  {\path{doi:10.1006/jcss.2000.1741}}.

\bibitem{LT10}
Tomasz {\L}uczak and St{\'e}phan Thomass{\'e}.
\newblock Coloring dense graphs via {VC}-dimension.
\newblock {\em arXiv preprint arXiv:1007.1670}, 2010.

\bibitem{manurangsi2017inapproximability}
Pasin Manurangsi and Aviad Rubinstein.
\newblock Inapproximability of {VC}-dimension and {L}ittlestone’s dimension.
\newblock In {\em Conference on Learning Theory}, pages 1432--1460. PMLR, 2017.

\bibitem{matouvsek1998geometric}
Ji{\v{r}}{\'\i} Matou{\v{s}}ek.
\newblock Geometric set systems.
\newblock In {\em European Congress of Mathematics: Budapest, July 22--26, 1996
  Volume II}, pages 1--27. Springer, 1998.

\bibitem{Mat99}
Ji{\v{r}}{\'i} Matou\v{s}ek.
\newblock {\em {VC}-Dimension and Discrepancy}, pages 137--169.
\newblock Springer Berlin Heidelberg, 1999.
\newblock \href {https://doi.org/10.1007/978-3-642-03942-3_5}
  {\path{doi:10.1007/978-3-642-03942-3_5}}.

\bibitem{BIOGRID}
Rose Oughtred, Chris Stark, Bobby-Joe Breitkreutz, Jennifer Rust, Lorrie
  Boucher, Christie Chang, Nadine Kolas, Lara O’Donnell, Genie Leung,
  Rochelle McAdam, et~al.
\newblock The {BioGRID} interaction database: 2019 update.
\newblock {\em Nucleic acids research}, 47(D1):D529--D541, 2019.

\bibitem{PT87}
Robert Paige and Robert~Endre Tarjan.
\newblock Three partition refinement algorithms.
\newblock {\em {SIAM} Journal on Computing}, 16(6):973--989, 1987.
\newblock \href {https://doi.org/10.1137/0216062} {\path{doi:10.1137/0216062}}.

\bibitem{papadimitriou1996limited}
Christos~H. Papadimitriou and Mihalis Yannakakis.
\newblock On limited nondeterminism and the complexity of the {VC} dimension.
\newblock {\em Journal of Computer and System Sciences}, 53(2):161--170, 1996.
\newblock \href {https://doi.org/10.1006/jcss.1996.0058}
  {\path{doi:10.1006/jcss.1996.0058}}.

\bibitem{DIP}
Lukasz Salwinski, Christopher~S. Miller, Adam~J. Smith, Frank~K. Pettit,
  James~U. Bowie, and David Eisenberg.
\newblock The database of interacting proteins: 2004 update.
\newblock {\em Nucleic acids research}, 32(suppl\_1):D449--D451, 2004.

\bibitem{DIMES}
Yuval Shavitt and Eran Shir.
\newblock {DIMES}: Let the internet measure itself.
\newblock {\em ACM SIGCOMM Computer Communication Review}, 35(5):71--74,
  October 2005.
\newblock \href {https://doi.org/10.1145/1096536.1096546}
  {\path{doi:10.1145/1096536.1096546}}.

\bibitem{CAIDA}
{The Cooperative Association for Internet Data Analysis ({CAIDA})}.
\newblock The {CAIDA AS} relationships dataset.
\newblock \url{http://www.caida.org/data/active/as-relationships/}, 2013.

\bibitem{doi:10.1137/1116025}
Vladimir~N. Vapnik and Alexey~Ya. Chervonenkis.
\newblock On the uniform convergence of relative frequencies of events to their
  probabilities.
\newblock {\em Theory of Probability \& Its Applications}, 16(2):264--280,
  1971.
\newblock \href {https://doi.org/10.1137/1116025} {\path{doi:10.1137/1116025}}.

\bibitem{welzl1988partition}
Emo Welzl.
\newblock Partition trees for triangle counting and other range searching
  problems.
\newblock In {\em Proceedings of the fourth annual symposium on Computational
  geometry}, pages 23--33, 1988.

\bibitem{wenocur1981some}
Roberta~S Wenocur and Richard~M Dudley.
\newblock Some special vapnik-chervonenkis classes.
\newblock {\em Discrete Mathematics}, 33(3):313--318, 1981.
\newblock \href {https://doi.org/10.1016/0012-365X(81)90274-0}
  {\path{doi:10.1016/0012-365X(81)90274-0}}.

\end{thebibliography}
